\documentclass[10pt,a4paper]{article}

\usepackage[utf8]{inputenc}
\usepackage[T1]{fontenc}
\usepackage{microtype}
\usepackage{fourier}

\usepackage{tikz}
\usepackage[margin=10pt,font=small,labelfont=bf,labelsep=period]{caption}
\usepackage{amsmath}
\usepackage{amsthm}
\usepackage{amssymb}
\usepackage{mathscinet}
\usepackage{pdfpages} 
\usepackage{empheq}
\usepackage{float}
\usepackage{subfig} 

\numberwithin{equation}{section}

\usepackage[nottoc,notlot,notlof]{tocbibind}

\usepackage{aliascnt}
\usepackage[colorlinks,linkcolor=blue,citecolor=blue]{hyperref}

\theoremstyle{plain}

\newtheorem{theorem}{Theorem}[section]

\newaliascnt{lemma}{theorem}
\newtheorem{lemma}[lemma]{Lemma}
\aliascntresetthe{lemma}

\newaliascnt{corollary}{theorem}

\aliascntresetthe{corollary}

\theoremstyle{definition}

\newaliascnt{definition}{theorem}
\newtheorem{definition}[definition]{Definition}
\aliascntresetthe{definition}

\newaliascnt{example}{theorem}
\newtheorem{example}[example]{Example}
\aliascntresetthe{example}

\newaliascnt{remark}{theorem}
\newtheorem{remark}[remark]{Remark}
\aliascntresetthe{remark}

\newaliascnt{assumption}{theorem}
\newtheorem{assumption}[assumption]{Assumption}
\aliascntresetthe{assumption}

\newaliascnt{proposition}{theorem}
\newtheorem{proposition}[proposition]{Proposition}
\aliascntresetthe{proposition}

\newcommand{\R}{\mathbf{R}}
\newcommand{\C}{\mathbf{C}}

\renewcommand{\epsilon}{\varepsilon}
 
\newcommand\Cl{\mathcal{C}\ell}

\newcommand{\abs}[1]{\left\lvert #1 \right\rvert}

\newcommand{\InMb}{\stackrel{\leftarrow}{\mathcal{L}_{M}}}
\newcommand{\InMf}{\stackrel{\rightarrow}{\mathcal{L}_{M}}}

\newcommand{\bo}{{\bf 1}}
\renewcommand{\l}{\lambda}

\DeclareMathOperator{\sgn}{sgn}

\title{ A generalization of the beam problem: Connection to multi-component Camassa–Holm dynamics }
\author{Richard Beals \thanks{Department of Mathematics, Yale University,  New Haven, CT 06520, USA;
E-mail: richard.beals@yale.edu}\and Jacek Szmigielski \thanks{Department of Mathematics \& Statistics and Centre for Quantum Topology and Its Applications (quanTA), University of Saskatchewan, Saskatoon, SK, S7N 5E6, CANADA; 
  E-mail: szmigiel@math.usask.ca} }

\date{\today}

\makeatletter
\hypersetup{%
  pdfauthor={Jacek Szmigielski},
  pdftitle={\@title},
}
\makeatother

\begin{document}

\maketitle
\begin{abstract} 

We extend the Euler–Bernoulli beam problem, formulated as a matrix string equation with a matrix-valued density, to a setting where the density takes values in a Clifford algebra, and we analyze its isospectral deformations. For discrete densities, we prove that the associated matrix Weyl function admits a Stieltjes-type continued fraction expansion with Clifford-valued coefficients. By mapping the problem from the finite interval to the real line, we uncover a direct link to a multi-component generalization of the Camassa–Holm equation. This yields a vectorized form of the Camassa–Holm equation invariant under arbitrary orthogonal group actions. As an illustration, we examine the dynamics of a two-atom (two-peakon) matrix measure in the special case of a Clifford algebra with two generators and Minkowski signature. Our analysis shows that, even when peakon waves remain spatially separated, they can engage in long-range, synchronized energy exchange.

\end{abstract}
\section{Introduction} 
The Euler-Bernoulli beam problem is a classic ordinary differential equation (ODE) problem:
\begin{equation}\label{eq:EBbeam}
D_x^2(r(x) D_x^2 \varphi)=\lambda^2 m(x) \varphi, \qquad -1<x<1,
\end{equation}
where $r(x)$ represents the flexural rigidity, and $m(x)$ is the mass density of the beam. Both are assumed to be positive functions or, more generally, measures. The spectral parameter $\lambda^2$ represents the square of the frequency. This problem admits various boundary conditions, as discussed, for example, in Barcilon \cite{barcilon-beam-royal}, corresponding to different configurations at the beam's endpoints. The beam problem is a fourth-order differential equation. Its second-order analogue, the string equation:
\begin{equation} \label{eq:string}
-D_x^2\varphi=\lambda^2 m(x) \varphi, \qquad -1<x<1,
\end{equation}
serves as a key example of the effectiveness of the separation of variables applied to the one-dimensional wave equation:
\begin{equation}
\frac{1}{c^2(x)}u_{tt}(x,t)=u_{xx}(x,t),
\end{equation}
where $\frac{1}{c^2(x)}=\frac{m(x)}{T}$, and $T$ denotes tension, which we take to be 1.  
These classical problems turned out to be closely related to a modern theory of shallow water waves, 
in particular, to the Camassa-Holm equation \cite{camassa-holm}
\begin{equation} 
m_t=(um)_x+u_xm, \qquad m=u-u_{xx}. 
\end{equation}
This connection was established in \cite{beals-sattinger-szmigielski:stieltjes}.  More recently, the present authors \cite{beals-szmigielski:2021p:2CH-euler-bernoulli-beam-noncommutative-continued-fractions} established an analogous connection for the beam problem \eqref{eq:EBbeam} and 
a two-component Camassa-Holm equation \footnote{{\color{blue} Unbeknownst to us, an equivalent system was proposed by Geng and Wang in \cite{geng-wang:Coupled-CH}}.}: 
\begin{align} 
&m_t=(um)_x+u_x m-vm, \qquad &n_t&=(un)_x+u_x n+vn, \label{eq:2CH}\\
&u_{xx}-4u=n+m,  \qquad &v_x&=n-m. \label{eq:2CHconstraints}
\end{align} 

The present note extends this relationship by introducing a family of generalizations labelled by orthogonal groups: we demonstrate that given a vector space $W$ and a bilinear, symmetric, non-degenerate form $B$ on $W$, it is possible to associate with this data a vectorization of the Camassa-Holm equation labelled by $B$, which remains invariant under the orthogonal group $O(W)$.  
This generalization has two equivalent formulations.  One  takes place on the finite interval, which we take to be 
$[-1,1]$  and the generalization leads to a family of isospectral deformations of the generalized Euler-Bernoulli beam problem.  The other formulation takes place on the infinite interval $(-\infty, \infty)$ and it gives a muli-component 
version of the Camassa-Holm equation.  The two-component system given by \eqref{eq:2CH}, \eqref{eq:2CHconstraints},  corresponds to $W=\R^2$ and $B$ having the Minkowski signature $(+,-)$.  
Moreover, all these generalizations support non-smooth solitons, usually referred to as \emph{peakons}.  
Peakons were originally proposed in the paper \cite{camassa-holm}.  In the formulation of Camassa and Holm
they are given by the ansatz: 
\begin{equation} \label{eq:upeakons} 
u(x,t)=\sum_{j=1}^N m_j(t) e^{-\abs{x-x_j(t)}}, 
\end{equation} 
where $m_j(t)$ and $x_j(t)$ are smooth functions of $t$.  They can be given a simple mechanical interpretation: 
$m_j$ are momenta, $x_j$ are positions, and they satisfy 
Hamilton's equations of motion
\begin{equation} 
\dot x_j=u(x_j), \qquad \dot m_j=-m_j u_x(x_j).  
\end{equation} 
They also have a dual (wave) interpretation: $m_j(t)$ is the amplitude of a wave with a sharp peak at 
$x_j(t)$, and the resulting wave $u(x,t)$ is a superposition of these individual profiles.  
We discuss a system later in the paper that generalizes Camassa-Holm peakons.

\begin{figure}[H]
\centering
\includegraphics[scale=0.34]{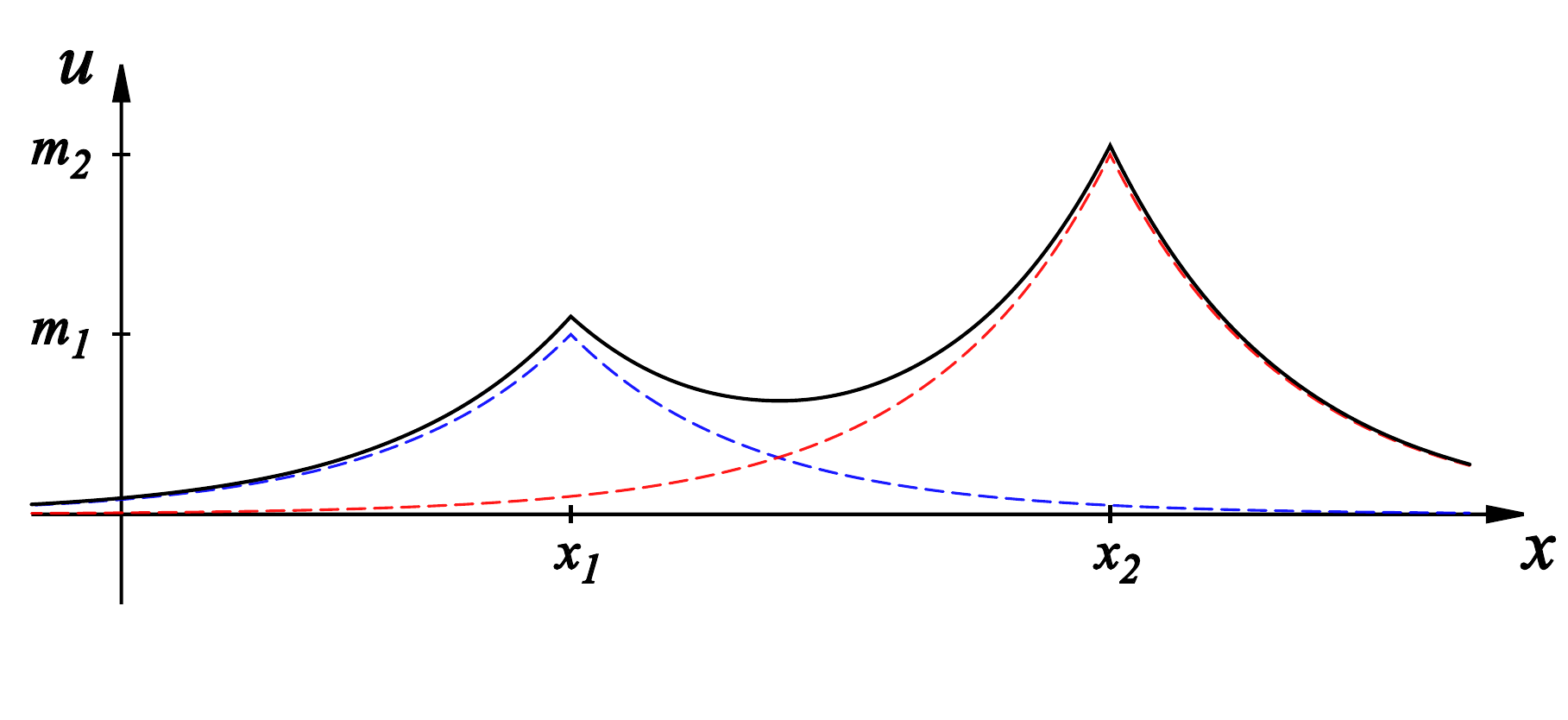}
\caption{Two-peakon solution}
\label{fig:2p}
\end{figure}
\section{ Euler-Bernoulli beam} 
The beam equation in the matrix form reads: 
\begin{equation}\label{eq:Mstring}
D_x^2 \phi=\lambda M \phi, \qquad -1<x<1, \quad M=\begin{bmatrix}0&n\\m&0 \end{bmatrix}, 
\phi=\begin{bmatrix} \phi_1\\ \phi_2 \end{bmatrix}. 
\end{equation} 
Let us review how this equation has been derived in \cite{beals-szmigielski:2021p:2CH-euler-bernoulli-beam-noncommutative-continued-fractions} from the Euler-Bernoulli beam problem $$D_x^2(r(x) D_x^2) \varphi=\lambda^2 m(x) \varphi, \qquad -1<x<1.  $$  Set
\begin{equation} 
\phi_1=\varphi, \qquad D_x^2 \varphi=\lambda n(x) \phi_2, \qquad \text{ where } n(x)=\frac{1}{r(x)},  
\end{equation} 
then 
\begin{equation} 
D_x^2 \phi_1=\lambda n(x) \phi_2, \qquad D_x^2 \phi_2=\lambda m(x) \phi_1, 
\end{equation} 
and the matrix equation \eqref{eq:Mstring} for $\phi=\begin{bmatrix} \phi_1\\ \phi_2\end{bmatrix}$ follows.  

Let us now consider the Dirichlet boundary value problem $\phi(-1)=\phi(1)=0$ and its general isospectral deformation.  
\begin{definition} 
\begin{equation}\label{eq:t-flow-differential form} 
D_t \phi=(\frac{a-b_x}{2} +b\, D_x)\phi,  
\end{equation}
where $a,b$ are $2\times 2$ matrix functions of $x$, a deformation parameter $t$,  and the spectral variable $\lambda$.  
\end{definition} 
Since we are deforming the Dirichlet boundary value problem, we require that $b(-1)=b(1)=0$ to preserve the boundary conditions.  
\begin{lemma}[Symmetric Form of compatibility conditions] \label{lem:ZCC}
\begin{subequations}
\begin{align}
\lambda M_t&=-\frac{b_{xxx}}{2}+\frac{\lambda}{2}\big(\InMf b+b\InMb+
[a,M]\big), \label{eq:t-flow}\\
0&=a_x+\lambda[b,M]. \label{eq:constraints} 
\end{align}
\end{subequations}
where $\InMf=D_x M+MD_x$ acting to the right and $\InMb=D_x M+MD_x$ acting to the left.  
\end{lemma} 

\begin{remark} The operator $D_xM+MD_x$ is common to all CH-type equations.  Writing it in the symmetric form 
reveals an important fact: the presence of anticommutators.  Indeed, $\InMf b+b\InMb=(Mb+bM)_x+Mb_x+b_xM$.  So, the deformation theory has both 
anticommutators and commutators (i.e.  the terms $[b,M]$ or $[a,M]$).  This is important for what follows.  In a significant  work on integrable systems on associative algebras 
Olver and Sokolov \cite{Olver-Sokolov:IS-associative} discuss the Hamiltonian structures for systems of this type, one of which appears on the right-hand side of \autoref{eq:t-flow}.\footnote{The authors thank A. Hone and V. Novikov for making us aware of this reference.}
\end{remark}
In \cite{beals-szmigielski:2021p:2CH-euler-bernoulli-beam-noncommutative-continued-fractions}, we considered the simplest ``negative'' flow: $$b=u(x,t)I +\frac{\beta(x)\sigma_1}{\lambda}, \quad a=v(x,t)\sigma_3, $$ where $\sigma_1, \sigma_3$ are Pauli matrices and 
$\beta(x)=1-x^2$.  To study the Dirichlet boundary value problem, we chose to work with the partial fundamental solution 
$\Phi(x,t)$ which is a $2\times 2$ matrix that satisfies 
\begin{equation} 
D_x^2 \Phi=\lambda M \Phi, \quad -1<x<1, \quad \quad \Phi(-1)=0, \Phi_x(-1)=\bo. 
\end{equation} 
Then, as is explained in \cite{beals-szmigielski:2021p:2CH-euler-bernoulli-beam-noncommutative-continued-fractions},  the Dirichlet spectrum $\mathcal{S}=\{\l \in \C: \det \Phi(+1,\lambda)=0\}$, in other words, 
the partial fundamental matrix is not invertible at $x=+1$.  
\subsection{Generalizing the Matrix String equation on the finite interval} \label{sec:Clifford beam}
Let $W$ be a  $d$-dimensional real vector space with a bilinear, symmetric form $B$.  We will initially assume that the bilinear form is non-degenerate with signature $(p,q), p+q=d$. 
The Clifford algebra $\mathcal{C\ell} (W,B)$ is generated by relations 
\begin{equation} \label{eq:Clifdef}
w_1w_2+w_2w_1=2B(w_1,w_2), \qquad w_1, w_2 \in W, 
\end{equation}
or, more concretely, we can choose an orthonormal, relative to $B$, basis $\{e_1,e_2, \cdots, e_d\}$ with $B(e_\mu,e_\nu)=\epsilon_\mu \delta_{\mu\nu}$ where the $\epsilon_\mu$ are $\pm 1$.  
We will drop the reference to the bilinear form $B$ in the notation for simplicity and write $B(u,v)=(u,v)$.  
In terms of the fixed orthonormal basis, the defining relations read
\begin{equation} \label{eq:Clifdef2}
e_\mu e_\nu+e_\nu e_\mu=2\epsilon_\mu \delta_{\mu \nu} , \qquad \mu, \nu  =1,2, \cdots, d. 
\end{equation} 
We recall that $\mathcal{C\ell} (W)$ splits as a vector space
\begin{equation} 
\mathcal{C\ell} (W)=\mathcal{C\ell}_0 (W)\oplus \mathcal{C\ell}_1 (W)\oplus \mathcal{C\ell}_2 (W)\oplus\cdots \oplus \mathcal{C\ell}_d(W). 
\end{equation} 
$\Cl_0 (W)$ is isomorphic to a copy of $\R$ while $\Cl_j(W)$ is spanned by $\binom{n}{j}$ elements $e_{\mu_1}e_{\mu_2}\cdots e_{\mu_j}$ for $\mu_1<\mu_2<\cdots<\mu_j$.  
$\Cl(W)$ is isomorphic as a vector space to the $2^d$ dimensional Grassmann algebra $\Lambda(W)=\bigoplus_{j=0}^n \Lambda ^j(W)$.  $\Cl_j(W)$ elements are assigned the (Grassmann) degree $j$.  
In this section, we make the following assumption about $M$:
\begin{assumption}
$M$ belongs to $\Cl_1$.  Thus 
\begin{equation} \label{eq:assM}
M(x,t)=\sum_{k=1}^d m_k(x,t) e_k.
\end{equation}
\end{assumption} 
\begin{theorem}
Let us consider the Dirichlet Matrix String problem with values in the Clifford algebra $\Cl(W)$
\begin{align} 
&D_x^2 \Phi=\lambda M \Phi, \qquad   -1<x<1,  \\
\Phi(-1, \l)=0,  \qquad &D_x\Phi(-1,\l)=1, \qquad \text{ and \hspace{0.1cm} $\Phi(+1,\lambda)$ is not invertible}, 
\end{align} 
for $M\in \Cl_1(W)$ and the family of isospectral deformations given by equations \eqref{eq:t-flow} and 
\eqref{eq:constraints}.  Then, there exists a solution to equations \eqref{eq:t-flow} and 
\eqref{eq:constraints}, of the form  
\begin{equation} \label{eq:EB}
b(x,t; \lambda)=b_0(x,t)+\frac{b_{-1}(x,t)}{\lambda}, \quad a(x,t; \lambda)=a(x,t), 
\end{equation} 
with $b_0(x,t) \in \Cl_0(W), \, b_{-1}(x,t)\in \Cl_1(W)$ and $a(x,t)\in \Cl_2(W)$.  More explicitly, if 
\begin{align}
&b_{-1}(x,t)=\beta(x) C=\beta(x)\sum_{\mu=1}^d c_\mu e_\mu,\qquad b_0(x,t)=u(x,t),\\
&      a(x,t)=\sum_{\mu<\nu} a_{\mu\nu }(x,t) e_\mu e_\nu, \hspace{3cm}
\end{align}
where $\beta(x)=1-x^2,$ $C$ is a fixed element in $\Cl_1(W)$, and whenever $\mu<\nu$, 
\begin{equation} \label{eq:vx,uxxx}
 \quad a_{\mu\nu,x}(x,t)=2\beta(x)\begin{vmatrix}m_\mu(x,t)&c_\mu\\ m_\nu(x,t)&c_\nu \end{vmatrix}, \quad u_{xxx}(x,t)=2(\InMf \beta(x), C), 
\end{equation} 
then $M(x,t)$ evolves according to 
\begin{equation} \label{eq:Cliffbeam}
M_t(x,t)=\InMf u(x,t)+\frac12[a,M], 
\end{equation} 
or, in components, 
\begin{equation} \label{eq:Cliffbeamcomp}
m_{\mu,t}(x,t)=\mathcal{L}_{m_\mu(x,t)}u(x,t)+\sum_{\nu=1}^d  a_{\mu\nu}(x,t)\epsilon_\nu m_\nu(x,t),   
\end{equation} 
 where  now $a_{\mu\nu}$ is the skew-symmetric extension of the $a_{\mu \nu}$ in \autoref{eq:vx,uxxx}.   
\end{theorem} 
\begin{proof} 
Using equations \eqref{eq:t-flow}, \eqref{eq:constraints} and \eqref{eq:EB} we get
\begin{equation} 
M_t=\frac12\big(\InMf b_0+b_0\InMb+[a,M]\big) \label{eq:Mt}, 
\end{equation}
and the constraints 
\begin{subequations}
\begin{equation}
b_{0,xxx}=\InMf b_{-1}+b_{-1}\InMb, \label{eq:b0}
\end{equation}
\begin{equation}
b_{-1,xxx}=0, \label{eq:b-1}
\end{equation}
\begin{equation}
[b_0, M]=0, \label{eq:b0M}
\end{equation}
\begin{equation}
a_{x}=[M,b_{-1}].  \label{eq:a0x}
\end{equation}
\end{subequations}

Setting $b_0=u(x,t)$ resolves the constraints \eqref{eq:b0M}.  Likewise, if we set $b_{-1}=\beta(x) C=\beta(x)\sum_{\mu=1}^d c_\mu e_\mu$ then \eqref{eq:b-1} holds.  
In the next step, we compute $\InMf b_{-1}+b_{-1}\InMb$ and $[M,b_{-1}]$ to show that \autoref{eq:vx,uxxx} hold.  Indeed, using \eqref{eq:Clifdef}, we obtain
\begin{equation} 
\begin{split} 
&\InMf b_{-1}+b_{-1}\InMb=(\beta(x)(M(x,t)C+CM(x,t)))_x+\beta_{x}(x)(M(x, t)C+CM(x,t))=\\
 &2B((M(x,t)\beta(x))_x,C)+2B(M(x,t)\beta_x(x),C)=2B(\InMf \beta(x),C).  
 \end{split}
\end{equation} 
Likewise, 
\begin{equation} 
\begin{split}
&[M,b_{-1}]=\beta(x)\sum_{\mu,\nu}m_\mu(x,t)c_\nu[e_\mu,e_\nu]=\\
&\beta(x)\sum_{\mu<\nu}(m_\mu(x,t)c_\nu-m_\mu(x,t) c_\mu)[e_\mu,e_\nu]=2\beta(x)\sum_{\mu<\nu} \begin{vmatrix} m_\mu(x,t)&c_\mu\\m_\nu(x,t)&c_\nu \end{vmatrix} e_\mu e_\nu. 
\end{split} 
\end{equation}
These two formulas, together with \eqref{eq:b0} and \eqref{eq:a0x},  prove \eqref{eq:vx,uxxx}.  Finally, we must compute $\frac12 [a,M]$.  
To this end, we note that 
\begin{equation} 
[e_\mu e_\nu, e_\rho]=2\delta_{\nu \rho} \epsilon_\rho e_\mu-2\delta_{\mu \rho}\epsilon_\rho  e_\nu.  
\end{equation} 
Thus 
\begin{equation} 
\begin{split}
&\frac12[a,M]=\frac12 \sum_{\mu<\nu} \sum_{\rho} a_{\mu \nu}(x,t)m_\rho(x,t)[e_\mu e_\nu, e_\rho]=\\
&\frac12 \big( \sum_{\mu<\nu} a_{\mu \nu}(x,t) 2m_\nu(x,t) \epsilon_\nu e_\mu-\frac12 \sum_{\mu<\nu}a_{\mu \nu}(x,t) 2m_\mu (x,t) \epsilon _\mu  e_\nu \big)=\\
&\sum_{\mu< \nu} a_{\mu \nu }(x,t) m_\nu(x,t)  \epsilon_\nu e_\mu- \sum_{\mu<\nu}a_{\mu \nu }(x,t) m_\mu (x,t) \epsilon_\mu  e_\nu=
\sum_\mu \sum_\nu  a_{\mu \nu}(x,t) m_\nu (x,t) \epsilon_\nu  e_\mu, 
\end{split}
\end{equation} 
with the proviso that $a_{\mu \nu}$, which originally is defined only for $\mu<\nu$,  is extended to a skew-symmetric tensor $a_{\mu \nu}(x,t)$ for all $\mu, \nu$.  Finally, the choice of $\beta(x)=1-x^2$ is dictated by two conditions: 
\begin{enumerate} 
\item $b_{-1},xxx=0$; 
\item $b_{-1}(-1)=b_{-1}(1)=0$  is needed to preserve the Dirichlet boundary conditions.  
\end{enumerate} 
\end{proof} 
It is instructive to write \autoref{eq:Cliffbeamcomp} as a vector equation.  
Let us denote by $\vec{M}=(m_1, m_2, \cdots, m_d), \vec{C}=(c_1, c_2, \cdots, c_d), S=\text{diag}(\epsilon_1, \epsilon_2, \cdots, \epsilon_d)$, the vector of masses, the constant vector $C$, and the signature matrix, respectively.  Then \autoref{eq:Cliffbeamcomp} can be written
\begin{equation} \label{eq:VCliff}
\vec{M}_t=\mathcal{L}_{\vec{M}} u+AS\vec{M}, 
\end{equation} 
where $A=[a_{\mu \nu}], \, \mathcal{L}_{\vec{M}}u=(\vec{M}u)_x+(\vec{M})_x u\, $,  and 
\begin{equation} \label{eq:Vconstraints1} 
\quad a_{\mu\nu,x}(x,t)=2\beta(x)\begin{vmatrix}m_\mu(x,t)&c_\mu\\ m_\nu(x,t)&c_\nu \end{vmatrix}, \quad u_{xxx}(x,t)=2(\mathcal{L}_{\vec{M}} \beta(x), \vec{C}).  
\end{equation}
We recall that the orthogonal group $O(W)$ acts on $\Cl(W)$, preserving $\Cl_1(W)$.  Thus 
it acts on $\vec{M}$.  

\begin{theorem} 
Equations \eqref{eq:VCliff} and \eqref{eq:Vconstraints1} are invariant under the action of 
the orthogonal group $O(W)$.  
\end{theorem} 
\begin{proof} Let $g\in O(W)$ and let $R_g$ be a matrix of rotation representing $g$.  
 Consider the action of $O(W)$ given by the following rules
\begin{equation*} 
\vec{M}\rightarrow  R_g\vec{M}, \quad 
\vec{C}\rightarrow R_g\vec{C}, \quad 
A\rightarrow R_gAR_g^T, \quad 
u\rightarrow u.  
\end{equation*} 
$R_g^T S R_g=S$ because the bilinear form $(,)$ is invariant under $O(W)$, and we get the claim.  
\end{proof}

\begin{example}
We consider $W=\R^2$.  The Clifford algebra $\Cl(W)$ is generated by $1, e_1, e_2$ where 
\begin{equation} 
e_1^2=\epsilon_1,\qquad e_2^2=\epsilon_2, \qquad \text{ and } e_1e_2+e_2e_1=0.  
\end{equation}
Let us denote $a_{12}=v$.  Then, the evolution equations \eqref{eq:Cliffbeamcomp} take the form
\begin{align}
&m_{1,t}=(m_1u)_x+u_xm_1+v\epsilon_2 m_2, \\
&m_{2,t}=(m_2u)_x+u_xm_2-v\epsilon_1 m_1,  
\end{align}
while the constraints \eqref{eq:vx,uxxx} are given by
\begin{equation} 
v_x=2\beta(x)\begin{vmatrix} m_1&c_1\\m_2&c_2 \end{vmatrix}, \quad u_{xxx}=2\mathcal{L}_{(M,C)}\beta(x).  
\end{equation} 
Consider now the special choice $c_1=1, c_2=0, \epsilon_1=1, \epsilon_2=-1$ and set $m=m_1+m_2, \, n=m_1-m_2$.  
Then 
\begin{align}
&m_t=\mathcal{L}_m u -v m, &n_t=\mathcal{L}_n u+vn, \\
&u_{xxx}=\mathcal{L}_{m+n} \beta, &v_x=\beta(n-m),   
\end{align}
which reproduces the original deformation equation of the Euler-Bernoulli beam put forward in \cite{beals-szmigielski:2021p:2CH-euler-bernoulli-beam-noncommutative-continued-fractions}.  Note that the Euler-Bernoulli beam is formally attached to the Minkowski signature $(+, -)$.  
\end{example} 
In the remainder of this section, we will discuss the issue of choosing a fixed $C$ as was done in the example 
above.  Fixing $C \in W$ splits $W$ into two orthogonal subspaces, i.e., 
\begin{equation}
W=V\oplus V^\perp, 
\end{equation} 
where $V^\perp=\R C$ and $V$ is its orthogonal complement.  
With the fixed $C$ the theory is now invariant under the orthogonal group $O(V)$ which leaves $V^\perp$ 
invariant.  We use the orthogonal projection $P_V$ on the subspace $V$ later in the paper.  
Also, for the remainder of the paper, we will choose $C=e_1$ following the choice used in the Euler-Bernoulli case.  This gives $V=<e_2,e_3, \cdots, e_d>$.  With this choice, the only non-zero component of $A$ is 
$a_{1\nu, x}=-2\beta(x) m_\nu(x,t), \nu=2,3, \cdots d$.  We now summarize the setup for this case.  

\begin{theorem}
Let us consider the Dirichlet Matrix String problem with values in the Clifford algebra $\Cl(W)$
\begin{align} 
&D_x^2 \Phi=\lambda M \Phi, \qquad   -1<x<1,  \\
\Phi(-1, \l)=0,  \qquad& D_x \Phi(-1, \l)=1, \quad  \text{ and \hspace{0.1cm} $\Phi(+1,\lambda)$ is not invertible}, 
\end{align} 
for $M\in \Cl_1(W)$ and the family of isospectral deformations given by equations \eqref{eq:t-flow} and 
\eqref{eq:constraints}.  Then, there exists a solution to equations \eqref{eq:t-flow} and 
\eqref{eq:constraints}, of the form  
\begin{equation} \label{eq:EBbis}
b(x,t; \lambda)=b_0(x,t)+\frac{b_{-1}(x,t)}{\lambda}, \quad a(x,t; \lambda)=a(x,t), 
\end{equation} 
where 
\begin{align}
&b_{-1}(x,t)=\beta(x) e_1, \qquad b_0(x,t)=u(x,t),\\
&      a(x,t)=e_1 v(x,t), \hspace{3cm}
\end{align}
where $\beta(x)=1-x^2,\,  v(x,t)\in \Cl_1(V)$, and 
\begin{equation} \label{eq:vx,uxxx bis}
 \quad v_x(x,t)=-2\beta(x)P_V M, \quad u_{xxx}(x,t)=2\mathcal{L}_{m_1}\beta(x) \epsilon_1, 
\end{equation} 
then $M(x,t)$ evolves according to 
\begin{equation} \label{eq:Cliffbeambis}
M_t(x,t)=\InMf u(x,t)+(M,v)e_1-(M, e_1)v, 
\end{equation} 
or, in vector form, 
\begin{equation} \label{eq:VCliffbis}
\vec{M}_t=\mathcal{L}_{\vec{M}} u+AS\vec{M}, 
\end{equation} 
where $A=\begin{bmatrix} 0&\vec{v}^T\\-\vec{v}&0 \end{bmatrix}$ and 
\begin{equation} \label{eq:Vconstraints} 
\quad \vec{v}_x=-2\beta(x)P_V \vec{M}, \quad u_{xxx}(x,t)=2\mathcal{L}_{m_1}\beta(x) \epsilon_1. 
\end{equation}
\end{theorem} 

\subsection{The Weyl function for the discrete Matrix String} 
In this section, we generalize the matrix-valued Weyl function used in the analysis of the beam problem in \cite{beals-szmigielski:2021p:2CH-euler-bernoulli-beam-noncommutative-continued-fractions} to the Clifford case.  The notation is that of \autoref{sec:Clifford beam}.  
We use the partial fundamental matrix $\Phi$, which satisfies 
\begin{equation} \label{eq:DIVP}
D_x^2 \Phi=\l M \Phi, \qquad -1<x<1, \qquad  \Phi(-1,\l)=0, \qquad D\Phi(-1, \l)=\bo. 
\end{equation} 
For simplicity, we will only consider the discrete case, which corresponds to peakons in the beam problem mapped to the real line (see next section). 
Thus, we assume 
that 
\begin{equation} 
M=\sum_{j=1}^N M_j \delta _{x_j}
\end{equation} where all $M_j\in \Cl_1(W)$ are assumed to be invertible.  We can explicitly construct the partial fundamental matrix $\Phi(x, \l)$ by constructing its restrictions to segments $[x_j, x_{j+1}]$.  Indeed, 
let us denote 
\begin{equation} 
\Phi_j=\Phi(x_j-)=\Phi(x_j), \qquad \Phi_j'=D\Phi(x_j-, \l),  \qquad l_j=x_{j+1}-x_j, \qquad 0\leq j\leq N, 
\end{equation} 
with the convention $x_0=-1, \, x_{N+1}=1$ and $\Phi_0=0, \: \Phi'_0=\bo$.  
Then on the interval $[x_j, x_{j+1}]$
\begin{equation} 
\Phi(x,\l)=\Phi_{j+1}'(x-x_j)+\Phi_j, 
\end{equation} 
and the jump condition
\begin{equation} 
\Phi_{j+1}'-\Phi_j'=\l M_j \Phi_j, 
\end{equation} 
which, equivalently, can be written as a recurrence relation
\begin{align} 
\Phi_{j+1}&=\Phi_j+\Phi_{j+1}'l_j, \label{eq:j+1}\\
\Phi_{j+1}'&=\Phi_j'+\l M_j \Phi_j,  \label{eq:prime j+1}.  
\end{align} 
The matrix form of the recurrence relation reads
\begin{equation} \label{eq:Matrecurrence}
\begin{bmatrix}\Phi_{j+1}\\ \Phi_{j+1}' \end{bmatrix}=\begin{bmatrix} 1&l_j\\0&1\end{bmatrix}\begin{bmatrix} 1&0\\\l M_j&1 \end{bmatrix} \begin{bmatrix} \Phi_j\\\Phi_j' \end{bmatrix}=\begin{bmatrix}1+\l l_j M_j&l_j\\\l M_j&1 \end{bmatrix} 
\begin{bmatrix} \Phi_j\\\Phi_j' \end{bmatrix} ,   \qquad 0\leq j\leq N, 
\end{equation} 
where we put $M_0=0$ for convenience.  
\begin{proposition}\label{prop:degrees}
\mbox{}\\
 For all $1\leq j\leq N+1$, the $\lambda$ degrees of $\Phi_j, \Phi_j'$  are: $\deg(\Phi_{j})=\deg(\Phi'_j)=j-1$.  

\end{proposition}
\begin{proof}
This follows directly from the recurrence relations.  
\end{proof}
\begin{proposition}\label{prop:inv}
For all $j\geq 1$, $\Phi_j$ and $\Phi'_j$ are invertible for sufficiently large
$\lambda$.  
\end{proposition}
\begin{proof}
The case $j=1$ is invertible, and the matrix recurrence relation \eqref{eq:Matrecurrence}  yields the proof by induction.  

\end{proof}
We can now summarize the recursive way of solving 
the initial value problem \eqref{eq:DIVP}.  
\begin{proposition}

\begin{equation}
\begin{bmatrix}\Phi_{j+1}\\ \Phi'_{j+1}\end{bmatrix}=
\left( \prod_{k=0}^{j}\begin{bmatrix}1+\lambda l_kM_k&l_k\\ \lambda M_k&1 \end{bmatrix}\right)\begin{bmatrix}0\\ 1\end{bmatrix}, \qquad \text{ for } 0\leq j\leq N.  
\end{equation}
\end{proposition}

\begin{theorem}[Continued fraction expansion; Stieltjes meets Clifford] \label{thm:SCfrac}
Let 
\begin{equation} 
W(\lambda)= \frac{1}{\l} \Phi_{N+1}'(1,\l)(\Phi_{N+1}(1,\l))^{-1}. 
\end{equation}
Then, for sufficiently large $\lambda$, $W(\lambda)$ has 
a Stieltjes continued fraction decomposition with 
Clifford valued coefficients
\begin{equation} \label{eq:SCfrac}
W(\lambda)=\cfrac{1}{\lambda l_N +\cfrac{1}{M_{N}+\cfrac{1}{\lambda l_{N-1}+\cfrac{1}{\ddots+\cfrac{1}{\lambda l_0}}}}}
\end{equation}
\end{theorem}
\begin{proof}
We note that by \autoref{prop:inv} all $\Phi_j$ and $\Phi'_j$ are invertible for sufficiently large $\lambda$.  By the last, $j=N$, recurrence relation (see \eqref{eq:j+1} and 
\eqref{eq:prime j+1}): 
\begin{equation*} 
\l \Phi_{N+1}(\Phi_{N+1}')^{-1}=\l \Phi_N(\Phi_{N+1}')^{-1}+\l l_N
\end{equation*}
and
\begin{equation*} 
\l^{-1} \Phi_{N+1}' (\Phi_{N+1})^{-1}=(\l l_N+(\frac1\lambda \Phi_N'(\Phi_N)^{-1}+M_N))^{-1}. 
\end{equation*} 

This is the first step in the iteration; we moved from $\frac1\l \Phi_{N+1}'(\Phi_{N+1})^{-1}$ to $\frac1\lambda \Phi_{N}'(\Phi_{N})^{-1}$.  We proceed by iterating down until we reach $\frac1\l \Phi_1'(\Phi_{1})^{-1}=\frac{1}{\l l_0}$, since
$\Phi_1=l_0$ and $\Phi'_1=1$ for the Dirichlet problem.   This concludes the proof.  
\end{proof}

\subsection{Generalizing the Matrix String equation to the infinite interval} 
We will briefly discuss the Clifford generalization from the previous section, adapted to the infinite interval $(-\infty, \infty)$.  These two pictures are connected via a Liouville transformation, which was initially introduced in 
the context of the CH equation in \cite{beals-sattinger-szmigielski:1998:acoustic-scattering-KdV-hierarchy}.  We refer to that paper for details regarding the transformation.  
These two pictures are equivalent (under suitable boundary conditions). Still, it makes sense to consider them separately because the generalization to the infinite interval is more directly connected to the CH equation.  \begin{definition} \label{def:alphaMstring}
Consider a compactly supported measure $M\in \Cl_1(V)$.  
The boundary value problem 
\begin{align*} 
&D_x^2 \Phi=(\alpha+\lambda M) \Phi, \qquad   -\infty<x<\infty, \quad \alpha >0,  \\
\Phi(x, \lambda) \rightarrow 0, \, \text{ as } x\rightarrow -\infty &,\, \text{ and $\Phi(x,\l)$ is non-invertible as } x\rightarrow \infty, 
\end{align*} 
will be called the \emph{alpha Matrix String}.  
  
\end{definition} 
We then consider isospectral deformations of the same type as in \eqref{eq:t-flow-differential form}, respecting the boundary conditions spelled out in 
\autoref{def:alphaMstring}.  
\begin{definition}\label{def:alphaMstring-Deformation}  
\begin{equation}\label{eq:alpha-t-flow-differential form} 
D_t \Phi=(\frac{a-b_x}{2} +b\, D_x)\Phi, 
\end{equation}
where $a,b\in \Cl(V) $ are functions of $x$, a deformation parameter $t$,  and the spectral variable $\lambda$, which are bounded as $\abs{x} \rightarrow \infty$.    
\end{definition} 
The alpha version of the compatibility condition stated in \autoref{lem:ZCC} is straightforward
\begin{lemma}[Symmetric Form of compatibility conditions; the alpha case]\label{lem:alphaZCC}
\begin{subequations}
\begin{align}
\lambda M_t&=-\frac{b_{xxx}}{2}+2\alpha b_x+\frac{\lambda}{2}\big(\InMf b+b\InMb+
[a,M]\big), \label{eq:alpha-t-flow}\\
0&=a_x+\lambda[b,M]. \label{eq:alpha-constraints} 
\end{align}
\end{subequations}
where $\InMf=D_x M+MD_x$ is acting to the right and $\InMb=D_x M+MD_x$ is acting to the left.  
\end{lemma}

\begin{theorem}[Clifford Matrix String on the infinite interval]
\mbox{}\\
Let us consider the boundary value problem of \autoref{def:alphaMstring} and 
the family of isospectral deformations specified in \autoref{def:alphaMstring-Deformation}.  
 Then, there exists a solution to equations \eqref{eq:alpha-t-flow} and 
\eqref{eq:alpha-constraints}, of the form  
\begin{equation} \label{eq:alphaEB}
b(x,t; \lambda)=b_0(x,t)+\frac{b_{-1}(x,t)}{\lambda}, \quad a(x,t; \lambda)=v(x,t), 
\end{equation} 
where 
\begin{align}
&b_{-1}(x,t)=e_1, \qquad b_0(x,t)=u(x,t),\\
&      a(x,t)=e_1 v(x,t), \hspace{3cm}
\end{align}
where $v(x,t)\in \Cl_1(V)$, and 
\begin{equation} \label{eq:vx,uxxx line}
 \quad v_x=-2P_V M, \quad u_{xxx}-4\alpha u_x=2m_{1, x}\epsilon_1, 
\end{equation} 
then $M(x,t)$ evolves according to 
\begin{equation} \label{eq:Cliff line}
M_t=\InMf u(x,t)+(M,v)e_1-(M, e_1)v, 
\end{equation} 
or, in vector form, 
\begin{equation} \label{eq:VCliff line}
\vec{M}_t=\mathcal{L}_{\vec{M}} u+AS\vec{M}, 
\end{equation} 
where $A=\begin{bmatrix} 0&\vec{v}\\-\vec{v}^T&0 \end{bmatrix}$ and 
\begin{equation} \label{eq:Vconstraints line} 
\quad \vec{v}_x=-2P_V \vec{M}, \quad u_{xxx}-4\alpha u_x=2m_{1, x}\epsilon_1. 
\end{equation}
\end{theorem} 
\begin{example}
We consider $V=\R^2$.  The Clifford algebra $\Cl(V)$ is generated by $1, e_1, e_2$ where 
\begin{equation} 
e_1^2=\epsilon_1,\qquad e_2^2=\epsilon_2, \qquad \text{ and } e_1e_2+e_2e_1=0.  
\end{equation}
The vector $\vec{v}$ has, in this case, one component.  Let us call it $v$.  Then, the evolution equations \eqref{eq:VCliff line} take the form
\begin{align}
&m_{1,t}=(m_1u)_x+u_xm_1+v\epsilon_2 m_2, \\
&m_{2,t}=(m_2u)_x+u_xm_2-v\epsilon_1 m_1,  
\end{align}
while the constraints \eqref{eq:Vconstraints line} are given by
\begin{equation} 
v_x=-2 m_2, \quad u_{xxx}-4\alpha u_x=2m_{1,x} \epsilon_1.  
\end{equation} 
Consider now the special choice $\epsilon_1=1, \epsilon_2=-1, \alpha=1$ and set $m=m_1+m_2, \, n=m_1-m_2$.  
Then 
\begin{align}
&m_t=\mathcal{L}_m u -v m, &n_t=\mathcal{L}_n u+vn, \\
&u_{xxx}-4 u_x=\mathcal{L}_{m+n} 1=(m+n)_x, &v_x=(n-m),   
\end{align}
in full agreement with the 2-component CH equation proposed in \cite{beals-szmigielski:2021p:2CH-euler-bernoulli-beam-noncommutative-continued-fractions}.  
\end{example} 

\section{Preliminaries about asymptotics for the alpha Matrix String with $\alpha=1$} 
We will set $\alpha=1$ for the rest of the paper.  
Since $M$ has compact support and $\Phi \rightarrow 0, \text{ as } x\rightarrow -\infty$, 
\begin{equation} \label{eq:Phiminus} 
\Phi(x,\lambda)=\mu e^x, \qquad x<< 0, \qquad \mu \in \Cl(W), 
\end{equation} 
where  $\mu \in \Cl(W)$, and $\mu$ can depend on the deformation parameter $t$.  We will discuss the significance of $\mu$ in the next section.  
For large $x$ 
\begin{equation} \label{eq:Phiplus}
\Phi(x, \lambda)=e^{x} A(\lambda)+e^{-x} B(\lambda), \qquad x>> 0, 
\end{equation} 
  holds, and, on the spectrum, $A(\lambda_\nu)$ is not invertible in $\Cl (W)$.  
Recall that the Lax pair for the alpha Matrix String, with $\alpha=1$,  is: 
\begin{subequations} 
\begin{align} 
\Phi_{xx}&=(1 +\lambda M) \Phi, \qquad M \in \Cl_1(W), \label{eq:x-Lax-DS} \\
\Phi_t&=\frac{a-b_x}{2}\Phi+b\Phi_x,  \label{eq:t-Lax-DS}
\end{align}
\end{subequations} 
where $\Phi, a$ and $b$ are functions of $x,t, \lambda$  with values in the Clifford algebra $\Cl(W)$.  In this work, we only consider the simplest deformation in which $a$ and $b$ have terms 
of degree $0, -1$ in $\lambda$.  As shown in the previous section, if we add the 
assumption that $a$ and $b$ are bounded  as $x\to\pm\infty$ and $C=e_1$, then,
up to a normalization, they take the form
\begin{align} \label{eq:abalphaMS}
&b(x,t,\l)=u(x,t)+\frac {e_1}{\l} ,\qquad u(x,t)\in \Cl_0, \\
\qquad 
&a(x,t)=e_1 v(x,t), \qquad v(x,t) \in \Cl_1(V).  
\end{align} 
The compatibility conditions split into two sets of constraints 
\begin{equation*} 
(u_{xx}-4 u)_x=2(M,e_1)_x,  \qquad v_x=-2P_V M, 
\end{equation*} 
and a system of evolution equations
\begin{equation} \label{eq:evolalphaMS} 
M_t=\mathcal{L}_M u+(v,M)e_1-(M,e_1) v. 
\end{equation} 
Furthermore, if we are interested in a compactly supported $M$ and $u$ vanishing at $\pm \infty$,  then we 
can replace one of the constraints with a more restrictive one,  while keeping the other constraint intact:  
\begin{equation} \label{eq:alphaconstraintsBCH}
u_{xx}-4 u=2(M,e_1) , \qquad v_{x}=-2P_V M.
\end{equation} 
Since $M$ has compact support,  \autoref{eq:alphaconstraintsBCH} can be solved explicitly to obtain 
\begin{align}
u(x,t)&=-\frac{1}{2} \int_{-\infty}^\infty e^{-2\abs{x-y}}(M(y,t), e_1) \, dy, \label{eq:uD} \\
v(x,t)&= -\int_{-\infty}^\infty \sgn (x-y) P_V M(y,t) \, dy \label{eq:vD}, 
\end{align} 
with the understanding that $v$ is determined only up to a constant element in $V$, giving the following asymptotic (in $x$) behaviour of $u$ and $v$: 
\begin{equation} \label{eq:uvassD}
\begin{matrix}\lim_{x\rightarrow \pm \infty} u(x,t)=0=\lim_{x\rightarrow \pm \infty}u_x(x,t), \\ 
\lim_{x\rightarrow \pm \infty}v(x,t)=\mp  \int_{-\infty}^\infty P_V M(y,t) dy\stackrel{\text{def}}{=}v_{\pm}(t).   
\end{matrix}\end{equation}

\begin{proposition} \label{prop:m1CL}
Suppose $M$ is compactly supported, $u$ is assumed to vanish at $\pm \infty$, and $v$ is chosen as in \autoref{eq:vD}.   Then the integral 
\begin{equation} \label{eq:MassConservation}
\int_{-\infty} ^\infty (M(x,t), e_1)\, dx
\end{equation} 
is independent of $t$.  
\end{proposition}
\begin{proof} 
Let $R$ be sufficiently large so that the interval $(-R,R)$ contains the support of $m_1$ (see \eqref{eq:assM}). Then, from 	\eqref{eq:evolalphaMS}, we obtain 
\begin{align*} 
D_t \int_{-\infty} ^\infty (M(x,t),e_1)\, dx=\int_{-\infty} ^\infty \mathcal{L}_{(M(x,t),e_1)}u(x, t)\, dx+
\int_{-R} ^R (v(x,t), M(x,t))\, dx\stackrel{\eqref{eq:alphaconstraintsBCH}}{=}\\
\int_{-\infty}^\infty\big(((M(x,t), e_1)u(x,t))_x+u_x\frac{u_{xx}-4u}{2}\big)\, dx -\int_{-R}^R \frac{(v,v)_x}{4} \, dx\stackrel{\eqref{eq:uvassD}}{=}0+0.  
\end{align*}
\end{proof} 
\begin{remark} We remark that in general $\int_{-\infty}^\infty (M(x,t), e_\nu) \, dx,\,  \nu=2,\cdots n$,  are not conserved.  
Indeed, the last term in  \eqref{eq:evolalphaMS} contributes 
$\int_{-\infty}^\infty \frac{u_{xx}-4u}{2} v \, dx=\frac12 \int_{\R} u_xv_x \, dx +2 \int_{\R}uv\neq 0$.  
We will, nevertheless, use $\int_\R M(x,t)\, dx$, which, other than the first component, is not conserved.  
Let us denote 
\begin{equation} \label{eq:calM}
\mathcal{M}(t)=\int_\R M(x,t) \, dx .  
\end{equation} 
\end{remark} 
Since $M$ is compactly supported, we have 
\begin{align}  
&u=u_-(t)e^{2 x}, \qquad &x<<0, \label{eq:asminus}\\
&u=u_+(t) e^{-2 x}, \qquad &0<< x, \label{eq:asplus}
\end{align} 
where 
$$ 
u_\pm(t)=-\frac{1}{2} \int_{-\infty}^\infty e^{\pm 2x}(M(x,t),e_1) \, dx. 
$$
We note that 
\begin{align} 
&\frac{a-b_x}{2}=\frac{e_1v_-}{2} - u_- e^{2 x} , &b&=u_-e^{2x}  +\frac {e_1}{\l },\qquad  &x<< 0, \label{eq:abminus} \\
&\frac{a-b_x}{2} = \frac{e_1v_+}{2}+u_+e^{-2x},  &b&=u_+e^{-2x} +\frac{e_1}{\l}, \qquad & x>> 0. \label{eq:abplus} 
\end{align} 
\section{Time evolution} 

For convenience, we recall the specific deformation used in the previous section: 
\begin{equation}\label{eq:D deform}
\Phi_t=\frac{(e_1 v-u_x)}{2} \Phi+ (u+\frac{e_1}{\lambda}) \Phi_x, \qquad v\in \Cl_1, \quad u\in \Cl_0,  
\end{equation} 
 where $\Phi$ satisfies 
\begin{equation} \label{eq:PhiDIVP}
 \Phi_{xx}=(1+\lambda M )\Phi, \qquad \Phi(x) \rightarrow 0, \qquad x\rightarrow -\infty. 
\end{equation}
 We argued in the previous section that the asymptotic behaviour of $\Phi$ as $x<<0$ is of the 
 form $\Phi(x,t,\l)=\mu(t) e^x$.  
 For the forward problem, it would be convenient to use the normalization $\mu(t)=1$, for which 
 \begin{equation} \label{eq:Phi normalization} 
 \Phi(x,t,\lambda)=e^x, \qquad \text{ for   } x<<0. 
 \end{equation}

However, $\Phi$ so normalized is not compatible with \eqref{eq:D deform}.  Indeed, 
in the asymptotic region $x<< 0$, the left hand of \eqref{eq:D deform} is $0$, while the righthand side is 
$(\frac{e_1 v_{-}(t)}{2} +\frac{e_1}{\l})e^x\neq 0$.  We note that \autoref{eq:PhiDIVP} is invariant under the right multiplication by a $t$-dependent invertible element from the Clifford Algebra $\Cl(W)$.   Taking advantage of that freedom, we are seeking a 
$t$-depended gauge transformation $\Omega(t, \l)$ such that $\Psi=\Phi \Omega$ can satisfy 
both \eqref{eq:PhiDIVP} and \eqref{eq:D  deform} with $\Phi$ is normalized as in \autoref{eq:Phi normalization}. 
This latter condition forces the deformation equation for $\Phi$ to be modified.  
\begin{lemma} Under the assymptotic constraint 
$\Phi(x,t,\l)=e^x$, imposed in the region $x<<0$, the deformation equation for $\Phi$ reads 
\begin{equation} \label{eq:t-DOmega} 
\Phi_t+\Phi \Omega_t \Omega^{-1} =\frac{(e_1 v-u_x)}{2} \Phi+ (u+\frac{e_1}{\lambda}) \Phi_x, 
\end{equation} 
with 
\begin{equation} \label{eq:Domegadef}
 \Omega_t \Omega^{-1} = \frac{e_1P_V \mathcal{M}(t)}{2} +\frac{e_1}{\l},    
\end{equation} 
and $\mathcal{M}(t)$ given by \autoref{eq:calM}.  
\end{lemma} 

\begin{proof}  \autoref{eq:D deform}, when written for $\Psi$, implies
\begin{equation*} 
\Phi_t+\Phi \Omega_t \Omega^{-1} =\frac{(e_1 v-u_x)}{2} \Phi+ (u+\frac{e_1}{\lambda}) \Phi_x
\end{equation*} 
Then, in the asymptotic region $x<< 0$ where $\Phi=e^x$, we 
get 
\begin{equation*}
 \Omega_t \Omega^{-1} =\frac{e_1 v_{-}(t)}{2} +\frac{e_1}{\l}\stackrel{\eqref{eq:calM}}{=} \frac{e_1P_V \mathcal{M}(t)}{2} +\frac{e_1}{\l} .   
\end{equation*}
\end{proof} 
The computation in the asymptotic region $x>>0$, where $\Phi(x,t, \l)=A(t,\l) e^x+B(t,\l) e^{-x}$,  can now be carried out.  
  
\begin{proposition} \label{prop:AB}
\begin{align} 
A_t=J_+A+AJ_-, \qquad & B_t=J_- B+BJ_-+2u_+ A, \label{eq:ABevol}\\
&(Ae_1)_t=[J_+, Ae_	1], \label{eq:Ae1evol}
\end{align} 
where 
\begin{equation} \label{eq:JpmD}
J_{\pm}(t,\l)=\frac{e_1v_{+}(t)}{2}\pm \frac{e_1}{\l}.
\end{equation} 
\end{proposition} 
\begin{proof} 
We use \eqref{eq:t-DOmega} and \eqref{eq:abplus} to obtain
\begin{align*} 
A _te^x+B_t e^{-x}+(Ae^x+Be^{-x})(\frac{e_1 v_-}{2}+\frac{e_1}{\l})=\\
(\frac{e_1 v_+}{2}+u_+ e^{-2x})(Ae^x+Be^{-x})+(u_+e^{-2x}+\frac{e_1}{\l})(Ae^x-Be^{-x}). 
\end{align*} 
Comparing the coefficients of $e^x$ and $e^{-x}$, along with utilizing \eqref{eq:uvassD}, provides 
\begin{align*}
&A_t=(\frac{e_1 v_+}{2}+\frac{e_1}{\l})A+A(\frac{e_1 v_+}{2}-\frac{e_1}{\l}), \\
&B_t=(\frac{e_1 v_+}{2}-\frac{e_1}{\l})B+B(\frac{e_1 v_+}{2}-\frac{e_1}{\l})+2u_+A,  
\end{align*}
which, by \eqref{eq:JpmD}, proves \eqref{eq:ABevol}.  Finally, since $v_+\in V, e_1 v_+=-v_+e_1$, we obtain  $J_-e_1=-e_1J_+$ thus proving  \eqref{eq:Ae1evol}.  
\end{proof} 
We define 
\begin{equation} \label{eq:defS}
S(t,\l)=A(t,\l)e_1
\end{equation} 
and conclude that, in view of \eqref{eq:Ae1evol}, $(S(t,\l),J_+(t,\l))$ form an algebraic Lax pair, satisfying
\begin{equation} \label{eq:LaxS}
S(t,\l)_t=[J_+(t,\l), S(t,\l)].  
\end{equation} 
We note that the Lax pair $(S(t,\l),J_+(t,\l))$ is a spectral version, depending only on $(t,\lambda)$,  of the original, $x$-dependent,  Lax 
pair $(D_x^2-\lambda M, D_t-(\frac{a-b_x}{2}+bD_x))$.

 \section{Finite beam} 
 \subsection{ Peakon equations} 
 Consider now the case of a finite, discrete measure
\begin{equation*} 
M=\sum_{j=1}^N M_j \delta_{x_j}, \quad x_1<x_2<\cdots<x_N, 
\end{equation*}
where $M_j\in \Cl_1 $.    We will only focus on a special case of invertible $M_j$.  
The system of PDEs \eqref{eq:evolalphaMS} and \eqref{eq:alphaconstraintsBCH}
reduces to a system of ODEs
\begin{align} 
\dot x_j&=-u(x_j), \label{eq:dotxj}\\
\dot M_j&=u_x (x_j)M_j +\frac12[e_1v(x_j), M_j],    \label{eq:dotMj} 
\end{align} 
$u_x(x_j), v(x_j)$ are the arithmetic averages of $u_x(x)$, $v(x)$, respectively, evaluated at the point of $x_j(t)$. 
We note that, since $u_x(x_j) $ and $ v(x_j)$ are continuous functions of $t$, the second equation implies that $M_j$ remains invertible under the time flow, provided $M_j(t=0)$ is invertible.  

The general form of the system \eqref{eq:dotxj} and \eqref{eq:dotMj} is not suitable for computational purposes.  We will now elaborate on 
two other, equivalent, presentations of the system \eqref{eq:dotxj} and \eqref{eq:dotMj}.   

First, with the help of \eqref{eq:uD} and \eqref{eq:vD}, we can display this system of equations in a more 
concrete fashion
\begin{align}
\dot x_j&=\frac12 \sum_{k=1}^N (M_k,e_1) e^{-2\abs{x_j-x_k}},\label{eq:dotxjbis} \\
\dot M_j&=\sum_{k=1}^N \sgn (x_j-x_k)\big((M_k, e_1)M_j  e^{-2\abs{x_j-x_k}}+(M_j, e_1)M_k-(M_j,M_k) e_1\big). \label{eq:dotMjbis}
\end{align}

Recall that by \autoref{prop:m1CL} the projection of the total mass $\mathcal{M}$ on $e_1$ is conserved.  Thus 
$$ 
(\mathcal{M}, e_1)=\sum_{j=1}^N (M_j, e_1)
$$ 
is constant; this result can be easily confirmed from \eqref{eq:dotMjbis}.     
The other way of writing these equations that is suitable for numerical computations is given by the 
following proposition.  
\begin{proposition} \label{prop:Clifford Peakons}
Let $m_{j\mu}=(M_j, e_\mu), \,\, 1\leq j\leq N, 1\leq \mu\leq d$.   Then \eqref{eq:dotxjbis} and \eqref{eq:dotMjbis}
are equivalent to 
\begin{empheq}[box=\fbox]{align*} 
&\dot x_j=\frac12 \sum_{k=1}^N m_{k1} e^{-2\abs{x_j-x_k}}, \\
& \dot m_{j1}=m_{j1} \sum_{k=1}^N \sgn(x_j-x_k)m_{k1}  e^{-2\abs{x_j-x_k}}-\epsilon_1\sum_{k=1}^N \sgn(x_j-x_k)\sum_{\mu=2}^d \epsilon_\mu m_{j\mu} m_{k\mu}, \\
&\dot m_{j\mu}=m_{j\mu}\sum_{k=1}^N \sgn(x_j-x_k)m_{k1} e^{-2\abs{x_j-x_k}}+m_{j1}\sum_{k=1}^N \sgn(x_j-x_k) m_{k\mu}, \quad 2\leq \mu\leq d.  
\end{empheq}

\end{proposition} 
\subsection{Peakon Sector; $N=2, d=2$} 
The remainder of the paper is about the original case inspired by the Euler-Bernoulli beam equation and its two-component CH counterpart.  
We work with the concrete matrix representation of $\Cl (\R^2)$ with signature $(+,-)$.  
Thus 
\begin{equation} 
e_1=\begin{bmatrix} 0&1\\1&0 \end{bmatrix}, \quad e_2=\begin{bmatrix} 0&-1\\1&0 \end{bmatrix}, \qquad 
e_1e_2=\begin{bmatrix} 1&0\\0&-1 \end{bmatrix}=\sigma_3, 
\end{equation} 
and $\Cl(\R^2)$ is spanned by $<\bo=I_2, e_1, e_2, e_1e_2>$.  
In this representation, $M$ is given by 
\begin{equation} 
M(t)=\sum_{j=1}^N M_j(t) \delta_{x_j(t)}, \qquad M_j(t)=\begin{bmatrix} 0&n_j(t)\\m_j(t)&0 \end{bmatrix} 
 , \quad m_j(0)>0, n_j(0)>0.  
\end{equation} 
In the case of $N=2$, we will write the peakon equations  implied by \autoref{prop:Clifford Peakons}, after a simple change of variables,  as 
\begin{proposition} \label{prop:twopeakons}
\begin{align} 
\dot x_1&=\frac14\big[(m_1+n_1)+(m_2+n_2) e^{-2\abs{x_1-x_2}}\big], \label{eq:x1}\\
\dot x_2&=\frac14\big[(m_1+n_1)e^{-2\abs{x_1-x_2}}+(m_2+n_2) \big], \label{eq:x2}\\
\dot m_1&=\frac{m_1}{2}\big[ -(m_2+n_2)e^{-2\abs{x_1-x_2}}-(m_2-n_2) \big], \label{eq:m1}\\
\dot n_1&=\frac{n_1}{2}\big[ -(m_2+n_2)e^{-2\abs{x_1-x_2}}+(m_2-n_2) \big], \label{eq:n1}\\
\dot m_2&=\frac{m_2}{2}\big[ (m_1+n_1)e^{-2\abs{x_1-x_2}}+(m_1-n_1) \big],  \label{eq:m2}\\
\dot n_2&=\frac{n_2}{2} \big[ (m_1+n_1)e^{-2\abs{x_1-x_2}}-(m_1-n_1)\big]. \label{eq:n2}
  \end{align} 
  \end{proposition} 

We recall that $\{S=Ae_1, J_+\}$ is the Lax pair for the deformed Euler-Bernoulli beam, mapped to the real axis,  in the asymptotic region $x>>0$ (see \eqref{eq:LaxS}).  We think of this pair as a \emph{spectral Lax pair} since it does not depend on $x$; it depends only on $(t,\l)$.  

We will prove elsewhere that 
 
\begin{equation*}
A=\mathbf{I}+\frac \lambda 2 (M_1+M_2)+ \big(\frac \lambda 2\big)^2 \big( 1-e^{-2(x_2-x_{1})}\big)M_2M_1,
\end{equation*}
which can be written explicitly as 
\begin{equation} 
A=\begin{bmatrix} 1+\big(\frac \lambda 2 \big)^2\big( 1-e^{-2(x_2-x_{1})}\big)n_2m_1& \frac \lambda 2(n_1+n_2)\\
\frac \lambda2 (m_1+m_2)& 1+\big(\frac \lambda 2\big)^2\big( 1-e^{-2(x_2-x_{1})}\big)m_2n_1\end{bmatrix}, 
\end{equation}
 giving the first constant of motion: 
 \begin{equation*}
I_1=n_1+n_2+m_1+m_2.
\end{equation*} 
Likewise, the determinant $\det(S)$, or, equivalently, $\det(A(\lambda))$,  is another source of constants of motion.  
The outcome of a simple computation is
\begin{align*} 
&\det(A(\l))=\\&1-\big(n_1m_1+n_2m_2 +(n_2m_1+m_2n_1)e^{-2(x_2-x_1)} \big) \big(\frac \lambda2\big)^2+ \big(1-e^{-2(x_2-x_1)}\big)^2 n_2m_2n_1m_1 \big(\frac\lambda2\big)^4,  
 \end{align*} 
 giving us two additional constants of motion: 
\begin{equation*} 
I_2=n_1m_1+n_2m_2 +(n_2m_1+m_2n_1)e^{-2(x_2-x_1)}  , \qquad I_3=\big(1-e^{-2(x_2-x_1)}\big)^2 n_2m_2n_1m_1.  
\end{equation*} 
In particular, since $m_j(0), n_j(0)>$ and the masses do not switch signs, the masses cannot collide since $I_3\neq 0$ and all masses are bounded.

Let us briefly explain the connection to the Riemann surface. This is important because the solutions exhibit periodic behaviour, which can be traced back to the Jacobian of that Riemann surface, although this periodic behaviour is not immediately apparent from the system in 
\autoref{prop:twopeakons}.  
The gist of the argument is simple: since $S$ satisfies the Lax equation \eqref{eq:LaxS} 
the Riemann surface $X$ of the affine curve 
$\det[ z\mathbf I -S]=0$ is time invariant; it is given explicitly by $
z^2=I_1 z+\det(A(\l)), $ or, after performing a simple affine transformation, 
\begin{center} 
\boxed{z^2=P_{4}(\lambda), } 
\end{center} 
where $P_4(\l)$ is a polynomial of degree $4$.  
We will show elsewhere that for $N$ peakons $z^2=P_{2N}(\l)$ and the genus of $X$ is 
generically $N-1$.  Thus, for two peakons, the associated Riemann surface $X$ has genus $1$.  

\subsection{Asymptotic behaviour} 
We discuss in this section the qualitative difference between peakons in the CH theory, which we will refer to as CH1, versus the peakons in the CH2 theory given by \eqref{eq:2CH} and \eqref{eq:2CHconstraints}.  
The detailed analysis will be presented in a separate publication, while here we only highlight the most striking features of CH2 that distinguish it from the CH1 theory.  
For simplicity, we will ignore proportionality constants and the differences in signs (the CH2 peakons in this paper have a negative sign in front (see \eqref{eq:uD}).  Thus will use 
\begin{equation} 
u_{CH1} (x,t)=m_1(t) e^{-2\abs{x-x_1(t)}} +m_2(t) e^{-2\abs{x-x_2(t)}}, 
\end{equation} 
for CH1 peakons and 
\begin{equation} 
u_{CH2} (x,t) = (m_1(t) + n_1(t)) e^{-2|x - x_1(t)|} + (m_2(t) + n_2(t)) e^{-2|x - x_2(t)|},
\end{equation} 
for CH2 peakons.  The dynamics of two peakons for the CH1 equation have been analyzed multiple times (see \cite{camassa-holm} or the recent review paper by one of us \cite{lundmark-szmigielski:review}).  In both cases, the wave function
$u(x,t)$ 
represents a superposition of two exponentially localized peaks. Each peak is centered at a time-dependent position \( x_1(t) \) or \( x_2(t) \), and has a time-dependent amplitude $m_1(t), m_2(t)$ for the CH1 and  \( m_1(t) + n_1(t) \) or \( m_2(t) + n_2(t) \), for CH2, respectively. These profiles decay exponentially away from their centers, forming a wave with two moving and interacting peaks (see \autoref{fig:2p}).   

For CH1, as one sees from \autoref{fig: CH1_figures}, the peaks initially interact. Then they separate and behave as if there were free particles moving with constant momenta, since both $m_1(t)$ and $m_2(t)$ become very quickly constant.  

The situation for CH2 peakons is sharply different.  Indeed, \autoref{fig: CH2_figures} 
shows that the peaks interact but then become approximately periodic after some transient behaviour.  This visualization highlights how the amplitudes of the two components evolve and interact. You can see differences in their transient behaviour and how they may settle into periodic or synchronized patterns.  The amplitudes \( m_1(t) + n_1(t) \) and \( m_2(t) + n_2(t) \) are not asymptotically constant; they oscillate over time. These oscillations are not in phase; one peak leads the other, creating a choreographed interaction. This behaviour suggests energy exchange between the two peaks, akin to two dancers passing momentum back and forth.  This is reminiscent of the concept of 
\emph{walzing peakons} discussed by Cotter, Holm, Ivanov and Percival in  \cite{Holm:Waltzing} occurring in a non-Lax integrable cross-coupled CH equation.  As time progresses, the CH2 system settles into a periodic regime. Both the positions and amplitudes of the peaks oscillate with a shared rhythm. This can be interpreted as a bound state or breather-like solution, where the wave maintains a dynamic equilibrium without dispersing.

In conclusion, the wave $ u_{CH2}(x,t) $  exhibits a rich interplay of motion and oscillation. It is both moving and dancing:
\begin{itemize}
  \item \textbf{Moving:} The peaks travel through space, indicating propagation.
  \item \textbf{Dancing:} The amplitudes oscillate, indicating internal dynamics and interaction.
\end{itemize}
This dual behaviour reflects a complex but structured evolution characteristic of nonlinear wave systems.  This behaviour of the CH2 equations \eqref{eq:2CH}, \eqref{eq:2CHconstraints}, as illustrated on an example of two peakons, can be described as ``moving and dancing'', where the wave peaks travel and their amplitudes oscillate in a choreographed manner.  A closely related phenomenon appears in resonant‐wave theory and soliton dynamics. For instance, in shallow‐water contexts, the three‐wave resonant interaction equations model spatially separated wave modes that continuously exchange energy in a phase‐locked fashion \cite{enwiki:three-wave}. Similarly, Onorato et al. (2003) observed modulational instability driving alternating energy transfer between spatially distinct shallow‐water wave packets \cite{Onorato:published}. An analogous effect of alternating energy beating between Fourier modes has also been demonstrated in coupled nonlinear theory, such as weakly interacting blue‑shifted modes in the nonlinear Schrödinger equation (e.g., \cite{griffits:instability}[Sect.2],  with near‑periodic exchanges). Even more strikingly, in optics, “breathing soliton molecules” are separated pulse pairs whose internal vibrations remain synchronized, with energy exchanged in lockstep \cite{wu:PhysRev}, much like our two-peakon system. These examples underscore how non‐dissipative, separated structures can remain dynamically coupled via conserved quantities or resonant interactions—exactly the mechanism we uncover in the two-peakon dynamics.

\begin{figure}[htbp]
\centering
\subfloat[Positions of peaks for two CH1 peakons]{\includegraphics[width=0.45\textwidth]{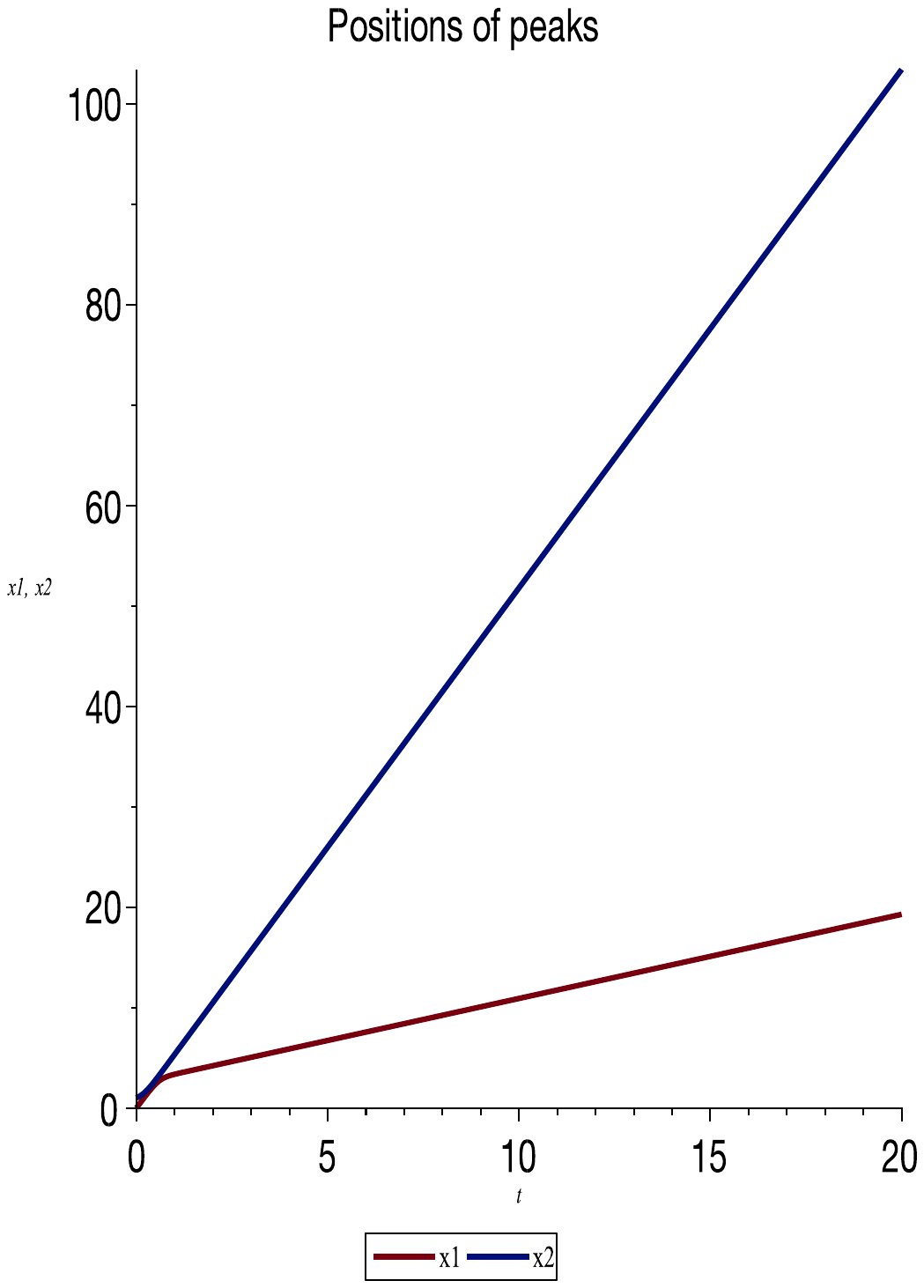}}
\hfill
\subfloat[Amplitudes of peaks for two CH1 peakons]{\includegraphics[width=0.45\textwidth]{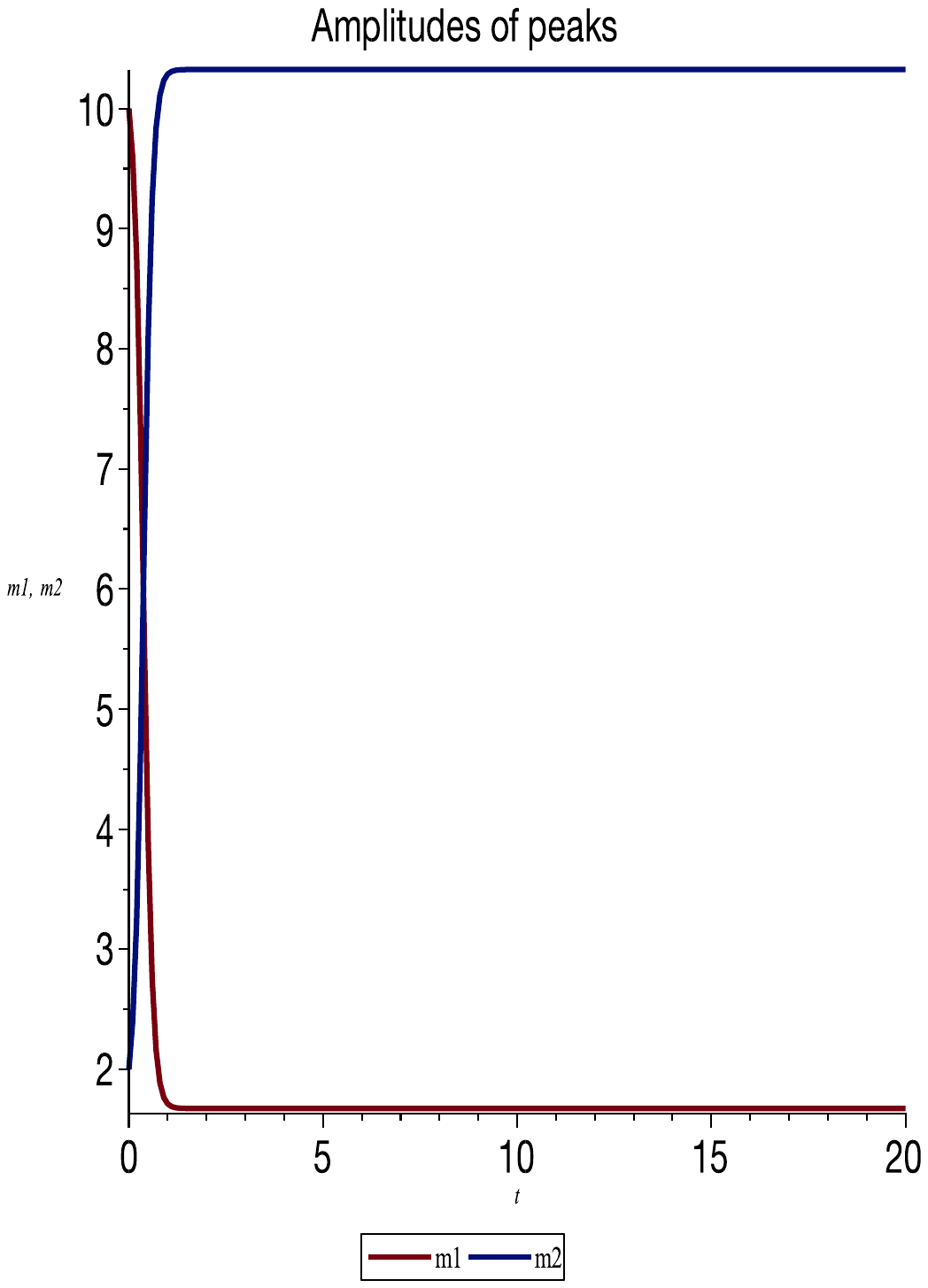}}
\caption{CH1 peakons}
\label{fig: CH1_figures}
\end{figure} 
\begin{figure}[htbp]
\centering
\subfloat[Positions of peaks for two CH2 peakons]{\includegraphics[width=0.45\textwidth]{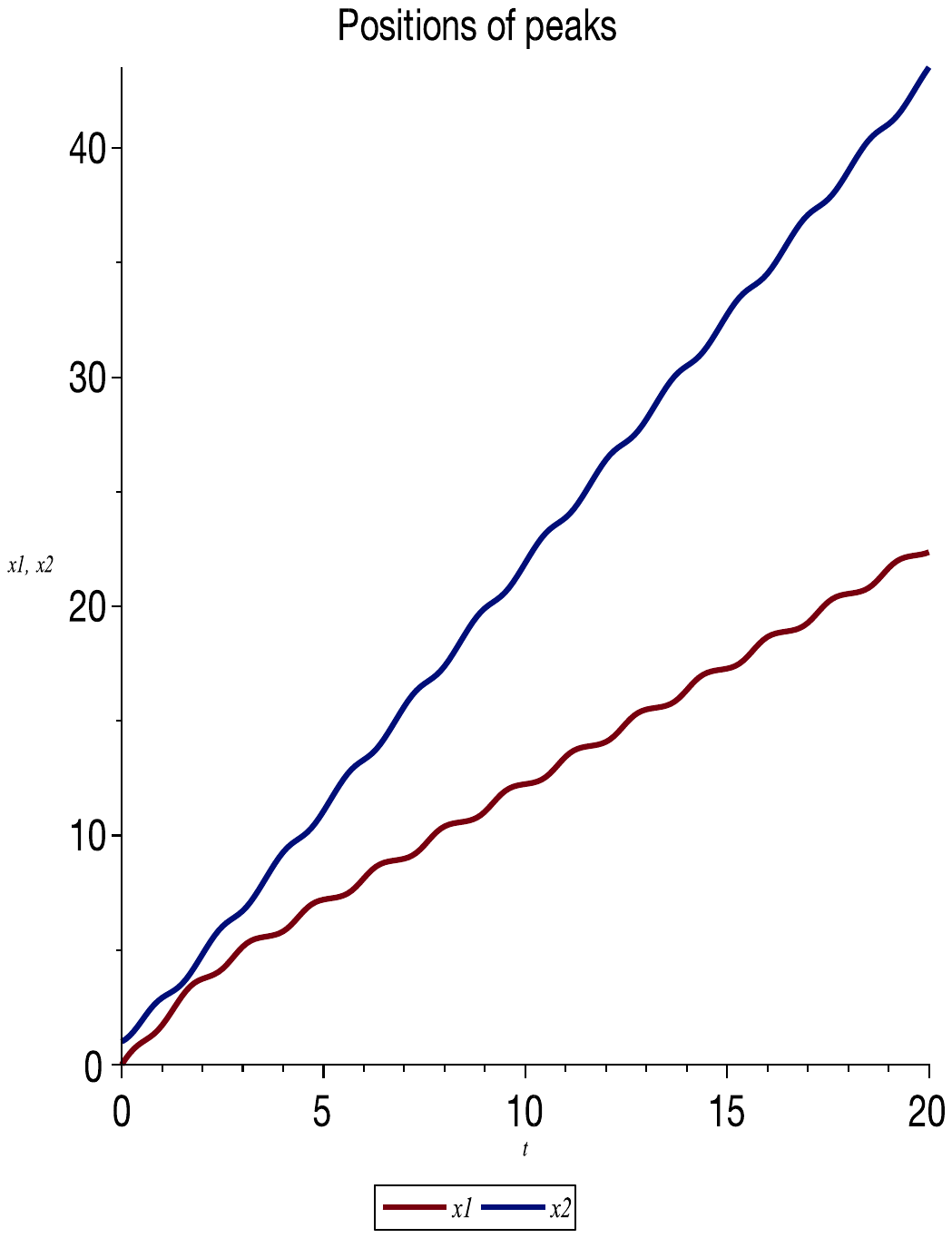}}
\hfill
\subfloat[Amplitudes of peaks for two CH2 peakons]{\includegraphics[width=0.45\textwidth]{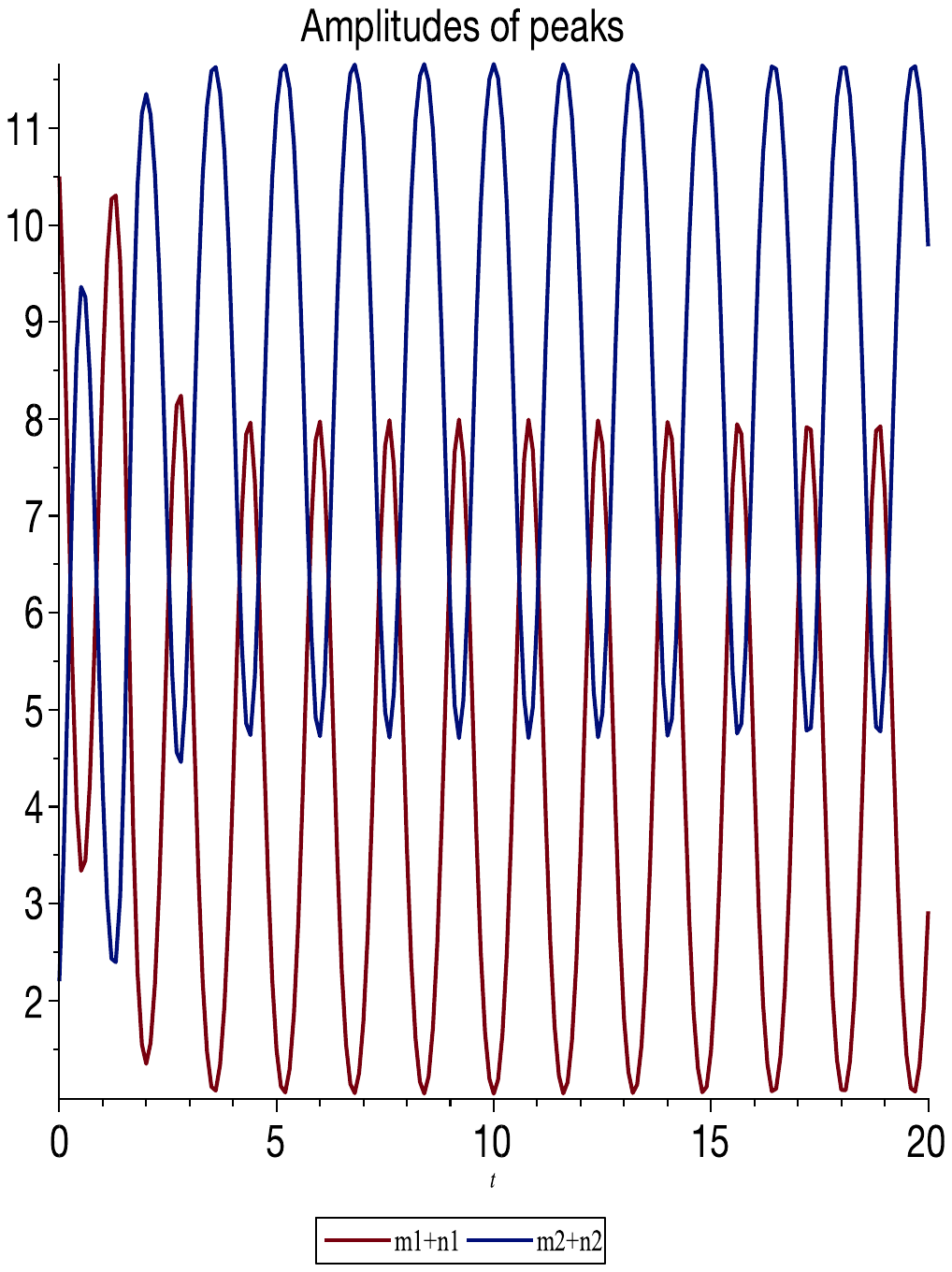}}
\caption{CH2 peakons}
\label{fig: CH2_figures}
\end{figure}

%

%

\newpage
\section{Acknowledgements}

Jacek Szmigielski's research is supported by the Natural Sciences and Engineering Research Council of Canada (NSERC).   We wish to thank 
A. Hone and V. Novikov for stimulating discussions and for bringing the reference to Olver and Sokolov \cite{Olver-Sokolov:IS-associative} to our attention.  We also thank X. Chang for his comments on the manuscript and 
S.Boscolo for consultation regarding \cite{wu:PhysRev}.  

\def\cydot{\leavevmode\raise.4ex\hbox{.}}
  \def\cydot{\leavevmode\raise.4ex\hbox{.}}



\begin{thebibliography}{10}

\bibitem{barcilon-beam-royal}
V.~Barcilon.
\newblock Inverse problem for the vibrating beam in the free--clamped
  configuration.
\newblock {\em Philosophical Transactions of the Royal Society of London.
  Series A, Mathematical and Physical Sciences}, 304(1483):211--251, 1982.

\bibitem{beals-sattinger-szmigielski:1998:acoustic-scattering-KdV-hierarchy}
R.~Beals, D.~H. Sattinger, and J.~Szmigielski.
\newblock Acoustic scattering and the extended {Korteweg}--de {Vries}
  hierarchy.
\newblock {\em Adv. Math.}, 140(2):190--206, 1998.

\bibitem{beals-sattinger-szmigielski:stieltjes}
R.~Beals, D.~H. Sattinger, and J.~Szmigielski.
\newblock Multi-peakons and a theorem of {S}tieltjes.
\newblock {\em Inverse Problems}, 15(1):L1--L4, 1999.

\bibitem{beals-szmigielski:2021p:2CH-euler-bernoulli-beam-noncommutative-continued-fractions}
R.~Beals and J.~Szmigielski.
\newblock A 2-component {C}amassa--{H}olm equation, {E}uler--{B}ernoulli beam
  problem, and noncommutative continued fractions.
\newblock {\em Comm. Pure Appl. Math.}, 76(10):2335--2371, 2023.

\bibitem{camassa-holm}
R.~Camassa and D.~D. Holm.
\newblock An integrable shallow water equation with peaked solitons.
\newblock {\em Phys. Rev. Lett.}, 71(11):1661--1664, 1993.

\bibitem{Holm:Waltzing}
C.~J. Cotter, D.~D. Holm, R.~I. Ivanov, and J.~R. Percival.
\newblock Waltzing peakons and compacton pairs in a cross-coupled
  {C}amassa--{H}olm equation.
\newblock {\em Journal of Physics A: Mathematical and Theoretical},
  44(26):265205, June 2011.

\bibitem{geng-wang:Coupled-CH}
X.~Geng and H.~Wang.
\newblock Coupled {C}amassa--{H}olm equations, {$N$}-peakons and infinitely
  many conservation laws.
\newblock {\em J. Math. Anal. Appl.}, 403(1):262--271, 2013.

\bibitem{griffits:instability}
S.~Griffiths, R.~Grimshaw, and K.~Khusnutdinova.
\newblock Modulational instability of two pairs of counter-propagating waves
  and energy exchange in a two-component system.
\newblock {\em Physica D: Nonlinear Phenomena}, 214(1):1--24, 2006.

\bibitem{lundmark-szmigielski:review}
H.~Lundmark and J.~Szmigielski.
\newblock A view of the peakon world through the lens of approximation theory.
\newblock {\em Phys. D}, 440:Paper No. 133446, 44, 2022.

\bibitem{Olver-Sokolov:IS-associative}
P.~J. Olver and V.~V. Sokolov.
\newblock Integrable evolution equations on associative algebras.
\newblock {\em Comm. Math. Phys.}, 193(2):245--268, 1998.

\bibitem{Onorato:published}
M.~Onorato, D.~Ambrosi, A.~R. Osborne, and M.~Serio.
\newblock Interaction of two quasi-monochromatic waves in shallow water.
\newblock {\em Physics of Fluids}, 15(12):3871--3874, 12 2003.

\bibitem{enwiki:three-wave}
{Wikipedia contributors}.
\newblock Three-wave equation --- {Wikipedia}{,} the free encyclopedia, 2025.
\newblock [Online; accessed 14-August-2025].

\bibitem{wu:PhysRev}
X.~Wu, J.~Peng, S.~Boscolo, C.~Finot, and H.~Zeng.
\newblock Synchronization, desynchronization, and intermediate regime of
  breathing solitons and soliton molecules in a laser cavity.
\newblock {\em Phys. Rev. Lett.}, 131:263802, Dec 2023.

\end{thebibliography}

\end{document}